\documentclass[letterpaper, 10pt, conference]{ieeeconf}
\IEEEoverridecommandlockouts
\usepackage{amsmath,amssymb}
\usepackage{mathtools}
\usepackage{bm}
\usepackage{amsthm}
\usepackage{cite}
\usepackage{graphicx}
\usepackage{subcaption}
\usepackage{booktabs}
\let\labelindent\relax
\usepackage{enumitem}
\usepackage{xcolor}
\usepackage{tikz}
\usepackage[hidelinks]{hyperref}
\usepackage{cleveref}

\crefname{theorem}{Theorem}{Theorems}
\crefname{lemma}{Lemma}{Lemmas}
\crefname{proposition}{Proposition}{Propositions}
\crefname{corollary}{Corollary}{Corollaries}
\crefname{definition}{Definition}{Definitions}
\crefname{assumption}{Assumption}{Assumptions}
\crefname{remark}{Remark}{Remarks}
\crefname{section}{Section}{Sections}
\crefname{subsection}{Section}{Sections}
\crefname{figure}{Fig.}{Figs.}
\crefname{table}{Table}{Tables}

\theoremstyle{plain}
\newtheorem{theorem}{Theorem}
\newtheorem{lemma}{Lemma}
\newtheorem{proposition}{Proposition}
\newtheorem{corollary}{Corollary}

\theoremstyle{definition}
\newtheorem{definition}{Definition}

\theoremstyle{remark}
\newtheorem{remark}{Remark}

\begin{document}

\title{\LARGE \bf Latent Safety Filters: When a Lossy Encoder Admits a Transferable Certificate}

\author{Johannes Mootz$^{1,2}$, Zahra Nili Ahmadabadi$^{3}$ and Reza Akhavian$^{1}$%
\thanks{$^{1}$Department of Civil, Construction, and Environmental Engineering, San Diego State University, San Diego, CA 92182, USA. {\tt\small \{jmootz, rakhavian\}@sdsu.edu}}%
\thanks{$^{2}$Department of Mechanical and Aerospace Engineering, University of California San Diego, La Jolla, CA 92093, USA.}%
\thanks{$^{3}$Department of Mechanical Engineering, San Diego State University, San Diego, CA 92182, USA. {\tt\small zniliahmadabadi@sdsu.edu}}%
}

\maketitle
\thispagestyle{empty}
\pagestyle{empty}

\begin{abstract}
Latent safety filters certify safety on a learned low-dimensional representation of the state, enabling constraints that resist analytic description.
Because the encoder is lossy, a filter can report safe while the physical state is unsafe, with no detectable model error.
Existing transfer conditions leave the effect of discarded safety information implicit.
We ask when a lossy encoder admits a safety certificate that transfers to the physical system, and show the answer is governed by the detectability of the discarded safety-relevant dynamics.
We construct a system whose latent model is exact and whose latent signals always report safe, while the physical state becomes arbitrarily unsafe.
For this system no certificate exists and no monitor downstream of the encoder can detect the failure.
When the discarded dynamics contract, a latent barrier certifies true safety up to two explicit margins, one for the latent-model error and one for the variation of safety across states the encoder cannot distinguish.
In the linear case and under boundedness and non-degeneracy conditions, every calibrated barrier transfers with a finite margin when the safety-relevant subspace is detectable, and none does otherwise.
On learned cartpole encoders, the model error does not indicate for which representations the estimated bound is non-vacuous, while the second margin does.
\end{abstract}

\section{Introduction}\label{sec:intro}

Latent safety filters replace the physical state with a learned representation and certify safety there.
This allows them to encode specifications that resist an analytical description, such as not spilling a container \cite{nakamura2025, kim2026latent}, or not overheating its contents \cite{kim2026latent}, working directly from high-dimensional observations.
They have been realized as Hamilton-Jacobi reachability filters \cite{nakamura2025}, as smooth latent control barrier functions (CBFs) \cite{nakamura2025train}, and as in-distribution barrier functions \cite{castaneda2023distribution}.
Extensions add out-of-distribution awareness \cite{seo2025uncertainty}, pre-trained vision backbones \cite{tabbara2025designing}, and runtime constraint adaptation \cite{agrawal2025anysafe}.
Often, encoder and latent dynamics are learned jointly as a world model, and the constraint is specified in its latent space \cite{nakamura2025,seo2025uncertainty}.

The guarantee must ultimately hold for the physical system and not for its representation.
Lutkus et al.\ give dynamics-aware conjugacy conditions and transfer a latent barrier or Lyapunov certificate back to the state, while charging a slack term for the latent-model error \cite{lutkus2025latent}.
Reduced-order models and simulation functions carry control barrier functions to a complex plant in the same way \cite{cohen2024safety,bousias2026transferring}.
In these works the measured quantity is the latent-model error, and the certificate holds with a conservative margin.
Reed et al.\ verify temporal-logic specifications on a finite abstraction of learned latent dynamics, with the encoder restricted to be convex so that the regions of interest are represented correctly \cite{reed2025scalable}.

However, a margin on the model error controls only the part of safety that the representation can see, and restricting the encoder class does not help with an encoder that is already trained.
A learned encoder compresses information and is lossy, so distinct physical states can collapse to a single representation $z$.
An encoder trained for reconstruction, or for the augmentation-invariance objectives common to DINOv2 backbones \cite{oquab2023dinov2}, does not necessarily preserve safety-relevant information.
A barrier certified on the latent space can then read safe while the physical state is unsafe. If the motion of the representation is reproduced exactly, it can do so with no detectable model error.

Kim et al.\ demonstrate such failures on hardware \cite{kim2026latent}.
They identify two failure modes: estimation gaps, where the current observation does not reveal a safety-critical quantity, and prediction gaps, where a failure is visible once it occurs but cannot be anticipated.
A mutual-information measure of safety observability diagnoses the first and a rollout-based measure of predictability the second, and the countermeasures are privileged multimodal supervision and conformal calibration.
Both diagnostics are empirical, and their closing question, when a learned representation is sufficient for control, is the one we formalize for safety.
Our conditions add what an information measure does not account for, the dynamics of the missing feature.
A missing feature that contracts costs a finite margin, while one that is hidden, unstable and unbounded in the safety variable $h$ admits no certificate, and no calibration downstream of the encoder can mitigate it, since no causal function of $z$ and the input $u$ separates the safe from the unsafe trajectory.

Wu et al.\ evaluate a latent certificate in the worst case over the encoder image of a state-uncertainty set, and uncover point-based certificates that report safe when they should not \cite{wu2026points}.
Their set is a bounded estimation error around a state estimate, while the fiber of a lossy encoder is fixed by the encoder, need not be bounded, and is not reduced by better sensing.

The prior literature leaves two questions open. The first is whether the representation retains enough of the safety variable $h$ for a latent certificate to stay valid \cite{kim2026latent}.
This is a weaker demand than reconstructing the full state.
The measurement-robust and output-feedback literature answers the related question of how a bound on estimation error becomes a safety margin \cite{dean2020measurement,cosner2021measurement,agrawal2022safe,chou2022safe}, but it assumes an explicit state estimate and not a learned representation.
The second question is which discarded directions are harmless to safety and which destroy it.
The geometric tools already exist: invariant foliations \cite{szalai2023data} and horizontal--vertical decompositions \cite{friedl2026reduced}.
For filters that certify an output of the physical state, this question is already answered. 
Such a filter leaves the internal dynamics unconstrained, and \cite{choi2026safety} gives minimum-phase conditions under which the full state stays bounded.
There the filter's own action drives a measured internal state unbounded, while the failure studied here occurs under zero input in a part of the state that the representation does not carry.

We show that detectability, in the sense of \cite{sontag1997output,hermann1977nonlinear}, answers the second question and governs whether a certificate can exist at all.
The fiber defect, the variation of true safety across the states the encoder cannot distinguish, answers the first question and measures what an existing certificate guarantees.
Our contributions are the following.

\begin{itemize}
\item We identify the discarded directions that defeat transfer.
A direction unobservable from the representation, unstable, and with true safety unbounded below along the fiber admits no certificate at all (\cref{thm:nec}).
No causal monitor of $(z,u)$ can detect the failure (\cref{cor:imp}).
\item We prove a transfer bound for the recoverable case (\Cref{thm:transfer}). A latent control barrier function bounds true safety from below, up to two explicit margins, one for the latent-model error and one for the fiber defect.
\item We give conditions under which the bound applies.
Contracting discarded dynamics keep the defect finite (\cref{lem:stable}), and in the linear case, under boundedness and non-degeneracy conditions, every calibrated barrier transfers with a finite margin when the safety-relevant subspace is detectable, and none does otherwise (\cref{prop:linear}).
\end{itemize}

\section{Problem Setup}\label{sec:setup}

We consider a closed-loop system on a domain $\mathcal D_x\subseteq\mathbb R^{n_x}$,
\[
  \dot x \;=\;f\bigl(x,\,\pi(E(x))\bigr)
\]
in which $f$ is the open-loop plant, the encoder $E:\mathbb R^{n_x}\!\to\mathbb R^{n_z}$ produces the representation $z=E(x)$, the controller $\pi:\mathbb R^{n_z}\!\to\mathbb R^{n_u}$ acts on $z$ alone, $u=\pi(z)$, and $f_z$ models how the representation evolves in latent space.
We write $F(x):=f(x,\pi(E(x)))$ for the closed-loop vector field.
We make no distinction between a nominal and a filtered policy.
When a latent safety filter modifies the input, $\pi$ denotes the filtered policy, which is again a function of $z$.
Throughout, $x$ is the physical state, on which dynamics and safety are defined, even when it is observed only through the encoder.
$R\subseteq\mathcal D_x$ denotes the closed-loop reachable set from the initial conditions of interest, which is forward invariant by construction.
Let $E$ be continuously differentiable and $F$ be locally Lipschitz with forward-complete trajectories on $\mathcal D_x$, so that $z(\cdot)$ and $u(\cdot)$ are defined for all $t\ge0$.

\subsection{Safety, model error and certificates}\label{sec:safety} 
Safety is specified at two levels: on the representation, as a continuously differentiable barrier $\phi:\mathbb R^{n_z}\!\to\mathbb R$, and on the state, as a true safety function $h:\mathbb R^{n_x}\!\to\mathbb R$, with safe set $\mathcal S_x=\{x:h(x)\ge0\}$.
The latent model $f_z$ need not be exact.
Following \cite{lutkus2025latent}, the error is measured by the forward-conjugacy bound: $f_z$ is \emph{$\gamma$-forward-conjugate} on $\mathcal D_x$ if
\[
  \sup_{x\in\mathcal D_x}\bigl\|\partial E(x)\,F(x)-f_z(E(x),\pi(E(x)))\bigr\|\le\gamma,
\]
so $\gamma$ upper-bounds the mismatch between the true motion of the representation and the model's prediction over the domain.
The residual $r(t)=\partial E(x(t))\,F(x(t))-f_z(z(t),\pi(z(t)))$ is the part of the representation's true motion that $f_z$ fails to predict.

\begin{definition}[Certificate]\label{def:cert}
A function $\phi:\mathbb R^{n_z}\!\to\mathbb R$ \emph{certifies} safety if its lift under-approximates true safety, $\phi(E(x))\le h(x)$ for every reachable $x$.
Then $\phi(E(x))\ge0$ implies $h(x)\ge0$.
\end{definition}

Write $\bar\phi:=\phi(E(x))$ for the \emph{lift}, a function of the state and constant along each fiber, so the comparison $\bar\phi\le h$ is well-posed.
The under-approximation is the standard requirement that a conservative surrogate not overstate safety \cite{ames2019cbf}.

\Cref{def:cert} is the function-valued counterpart of the set-inclusion view of latent transfer \cite{lutkus2025latent}.
There the induced safe set $\mathcal S_z=(E(\mathcal S_x^c))^c$ collects the representations no unsafe state maps to, the inclusion $E^{-1}(\mathcal S_z)\subseteq\mathcal S_x$ guarantees that latent-safe states are truly safe, and the inequality $\bar\phi\le h$ quantifies how safe.

\begin{definition}[Fibers and indistinguishability]\label{def:indist}
The \emph{fiber} over a representation $z$ is $E^{-1}(z)$, the set of states that share it.
Two initial states are \emph{indistinguishable} if the closed-loop trajectories they generate produce the same representation $z(\cdot)$, and hence the same input $u(\cdot)=\pi(z(\cdot))$, for all $t\ge0$.
This is the notion of \cite{hermann1977nonlinear} for the closed loop under the fixed policy $\pi$.
Their notion for the open-loop plant asks for equal outputs under every admissible input and is stronger, so two states remain indistinguishable here if only an input the policy never applies would separate them.
\end{definition}

\subsection{A motivating example}\label{sec:motivating}

Let the state be $x = (x_1, x_2)$, let the encoder keep only the first coordinate, $E(x) = x_1$, and let the dynamics be
\begin{equation}
  \dot{x}_1 = u + c\,x_2, \qquad
  \dot{x}_2 = b\,x_1 - a\,x_2,
  \label{eq:motivating-family}
\end{equation}
with $c = 0$ unless stated otherwise: $a > 0$ makes the hidden coordinate stable, $a < 0$ unstable.
The representation is $z = x_1$, so the fibers are the vertical lines $\{x_1 = \mathrm{const}\}$ and the encoder is blind to position along $x_2$.
Because any certificate is constant along each fiber, it must under-approximate true safety uniformly over the reachable states the encoder confuses:
\begin{equation}
  \phi(z) \;\le\; \inf\bigl\{\, h(x) : x \in E^{-1}(z) \cap R \,\bigr\}.
  \label{eq:fiber-inf}
\end{equation}
How restrictive this is depends on the safety function.
If $h$ depends only on $x_1$, as in $h=1-x_1$, then $h$ is constant along every fiber, the representation already determines safety, and the discarded coordinate is irrelevant (\cref{fig:fiber-a}).
If instead $h$ depends on $x_2$, as in $h=(1-x_1)-\ell\,x_2$ with $\ell>0$, a fiber can cross the boundary $\{h=0\}$ (\cref{fig:fiber-b}), and three behaviors are possible.
For bounded $x_1$, a stable $x_2$ ($a>0$) is limited to a bounded band.
An $x_2$ that influences the representation ($c\neq0$) appears in the evolution of $z$, and recovery is an observer problem, where the direct reconstruction $x_2=(\dot z-u)/c$ has gain $1/c$.
Only an $x_2$ at once unstable and invisible to $z$ is unrecoverable at runtime, and even then only if $h$ keeps falling as $x_2$ grows.
A fall that saturates, as in \cref{fig:fiber-c}, still admits a conservative certificate.
In the unrecoverable case $c=0$ and $\dot z=u$ exactly, so the discarded direction leaves no trace in the model error and the failure is invisible to any runtime check on the model.
\Cref{fig:triage} organizes the possible outcomes.

\begin{figure}[!t]
\centering

\definecolor{paperblue}{RGB}{0,114,189}
\definecolor{paperred}{RGB}{191,38,38}
\definecolor{paleblue}{RGB}{235,245,252}
\definecolor{palegray}{RGB}{242,242,242}
\definecolor{refgray}{RGB}{115,115,115}
\definecolor{textgray}{RGB}{70,70,70}

\begin{subfigure}{0.32\columnwidth}\centering
\begin{tikzpicture}[scale=0.62]
  \useasboundingbox (-0.1,-0.6) rectangle (3.75,3.2);
  \fill[paleblue] (0,0) rectangle (1.8,2.4);
  \fill[palegray] (1.8,0) rectangle (3.0,2.4);

  \draw[->] (0,0)--(3.2,0) node[right,font=\scriptsize]{$x_1$};
  \draw[->] (0,0)--(0,2.6) node[above,font=\scriptsize]{$x_2$};

  \draw[paperred,thick] (1.8,0)--(1.8,2.4)
    node[above,paperred,font=\scriptsize]{$h=0$};

  \draw[densely dashed,semithick,refgray] (1,0)--(1,2.4);
  \node[rotate=90,font=\scriptsize,text=textgray] at (0.68,1.6) {fiber};
  \draw[->,thick,paperblue] (1,0.8)--(1,1.9);

  \node[font=\scriptsize,text=textgray] at (0.48,0.45) {safe};
  \node[font=\scriptsize,text=textgray] at (2.4,1.9) {unsafe};

  \fill (1,0) circle (1.3pt) node[below,font=\scriptsize]{$z$};
\end{tikzpicture}
\caption{}\label{fig:fiber-a}
\end{subfigure}\hfill
\begin{subfigure}{0.32\columnwidth}\centering
\begin{tikzpicture}[scale=0.62]
  \useasboundingbox (-0.1,-0.6) rectangle (3.75,3.2);
  \fill[paleblue] (0,0)--(0,2.2)--(2.8,0)--cycle;
  \fill[palegray] (0,2.2)--(0,2.4)--(3.0,2.4)--(3.0,0)--(2.8,0)--cycle;

  \draw[->] (0,0)--(3.2,0) node[right,font=\scriptsize]{$x_1$};
  \draw[->] (0,0)--(0,2.6) node[above,font=\scriptsize]{$x_2$};

  \draw[paperred,thick] (0,2.2)--(2.8,0)
    node[pos=0.70,above right,paperred,font=\scriptsize,inner sep=1.5pt]{$h=0$};

  \draw[densely dashed,semithick,refgray] (1,0)--(1,2.4);

  \draw[->,thick,paperblue] (1,0.8)--(1,1.9);

  \node[font=\scriptsize,text=textgray] at (0.48,0.45) {safe};
  \node[font=\scriptsize,text=textgray] at (2.4,1.9) {unsafe};

  \fill (1,0) circle (1.3pt) node[below,font=\scriptsize]{$z$};
\end{tikzpicture}
\caption{}\label{fig:fiber-b}
\end{subfigure}\hfill
\begin{subfigure}{0.32\columnwidth}\centering
\begin{tikzpicture}[scale=0.62]
  \useasboundingbox (-0.1,-0.6) rectangle (3.75,3.2);
  \fill[paleblue]
    (0,2.4)--(0,0)--(2.6,0)
    -- plot[smooth] coordinates
       {(2.6,0)(2.25,0.3)(2.0,0.6)(1.8,1.0)(1.7,1.5)(1.66,2.0)(1.64,2.4)}
    -- cycle;
  \fill[palegray]
    plot[smooth] coordinates
       {(2.6,0)(2.25,0.3)(2.0,0.6)(1.8,1.0)(1.7,1.5)(1.66,2.0)(1.64,2.4)}
    -- (3.0,2.4)--(3.0,0)--cycle;

  \draw[->] (0,0)--(3.2,0) node[right,font=\scriptsize]{$x_1$};
  \draw[->] (0,0)--(0,2.6) node[above,font=\scriptsize]{$x_2$};

  \draw[paperred,thick]
    plot[smooth] coordinates
    {(2.6,0)(2.25,0.3)(2.0,0.6)(1.8,1.0)(1.7,1.5)(1.66,2.0)(1.64,2.4)};
  \node[paperred,font=\scriptsize,anchor=west] at (1.90,1.10) {$h=0$};
  \draw[paperred,densely dotted] (1.6,0)--(1.6,2.4);
  \node[paperred,font=\scriptsize,below,text height=1ex,text depth=0pt] at (1.6,0) {$1{-}\ell$};

  \draw[densely dashed,semithick,refgray] (1.0,0)--(1.0,2.4);

  \draw[->,thick,paperblue] (1.0,0.8)--(1.0,1.9);

  \node[font=\scriptsize,text=textgray] at (0.48,0.45) {safe};
  \node[font=\scriptsize,text=textgray] at (2.4,1.9) {unsafe};

  \fill (1.0,0) circle (1.3pt) node[below,font=\scriptsize]{$z$};
\end{tikzpicture}
\caption{}\label{fig:fiber-c}
\end{subfigure}

\caption{The dashed line is the fiber $\{x_1=\text{const}\}$ over $z$, and the blue arrow is a trajectory that moves up the fiber while $z$ stays fixed.
(\subref{fig:fiber-a})~$h$ is constant along each fiber, so the trajectory is harmless.
(\subref{fig:fiber-b})~The same trajectory leaves the safe set.
(\subref{fig:fiber-c})~The fall of $h$ saturates, so a finite certificate remains possible.}
\label{fig:fibers}
\end{figure}
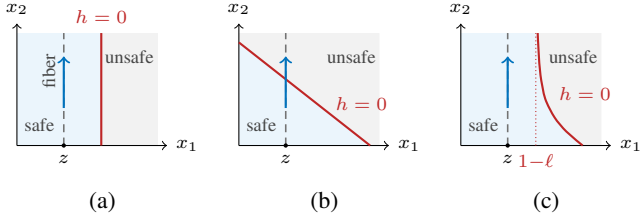

\begin{figure}[t]
\centering
\begin{tikzpicture}[
  q/.style={draw, rounded corners=2pt, align=center, inner sep=2.5pt,
            text width=3.8cm, font=\footnotesize},
  leaf/.style={draw, thick, rounded corners=2pt, align=center,
            inner sep=2.5pt, text width=3.5cm, font=\footnotesize},
  lab/.style={font=\scriptsize\itshape, inner sep=1.5pt, fill=white},
  >=stealth]
  \node[q] (q1) at (0,0)
    {Is the discarded direction\\ safety-relevant?\\[1pt]
     {\scriptsize $h$ not constant on the\\ reachable fiber
      (\cref{def:unbdd})}};
  \node[q] (q2) at (0,-1.95)
    {Is it hidden from the representation?\\[1pt]
     {\scriptsize indistinguishable in the $(z,u)$-history
      (\cref{def:indist})}};
  \node[q] (q3) at (0,-3.9)
    {Is it unboundedly safety-relevant?\\[1pt]
     {\scriptsize $h$ unbounded below on a reachable fiber
      (\cref{def:unbdd})}};
  \node[leaf] (l1) at (4.65,0)
    {\textbf{Exact transfer}\\[1pt]
     {\scriptsize the representation determines safety, up to the model-error margin}};
  \node[leaf] (l2) at (4.65,-1.95)
    {\textbf{Observer regime}\\[1pt]
     {\scriptsize bounded defect in the linear case
      (\cref{prop:linear})}};
  \node[leaf] (l3) at (4.65,-3.9)
    {\textbf{Transfer with floor}\\[1pt]
     {\scriptsize certificate transfers at the cost of a computable margin (\cref{thm:transfer}, V-A, V-C)}};
  \node[leaf] (l4) at (4.65,-5.55)
    {\textbf{No certificate, no monitor}\\[1pt]
     {\scriptsize the failure is invisible downstream of the encoder (\cref{thm:nec}, \cref{cor:imp}, V-A, V-B)}};
  \draw[->] (q1) -- node[lab]{no} (l1);
  \draw[->] (q1) -- node[lab]{yes} (q2);
  \draw[->] (q2) -- node[lab]{no} (l2);
  \draw[->] (q2) -- node[lab]{yes} (q3);
  \draw[->] (q3) -- node[lab]{no} (l3);
  \draw[->] (q3) |- node[lab, pos=0.2]{yes} (l4);
\end{tikzpicture}
\caption{Overview of the main results.}
\label{fig:triage}
\end{figure}
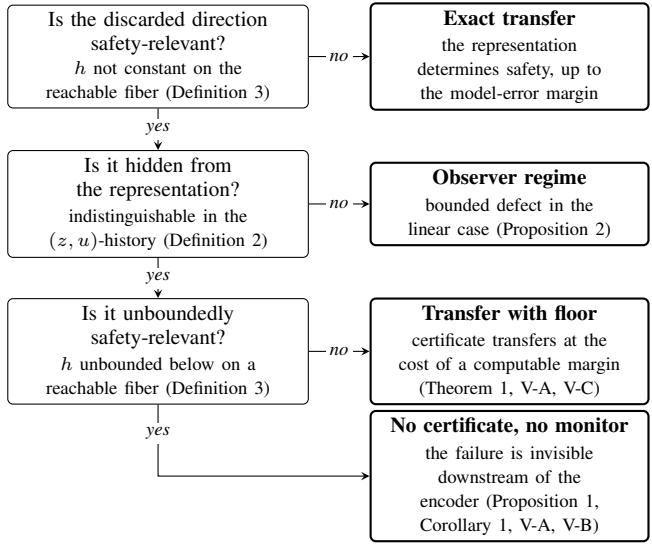

\section{Impossibility of Transfer}\label{sec:limit}

\subsection{Unrecoverable directions}\label{sec:condition}
We first isolate the kind of discarded direction that is unrecoverable.

\begin{definition}[Unboundedly safety-relevant]\label{def:unbdd}
The discarded information is \emph{safety-relevant} on the fiber over $z$ if $h$ is not constant on $E^{-1}(z) \cap R$, and \emph{unboundedly safety-relevant} on the fiber over $z$ if $h$ is unbounded below on its reachable part,
\[
  \inf_{x\in E^{-1}(z)\cap R}h(x)=-\infty ,
\]
so that no finite value under-approximates true safety on the reachable states the encoder collapses to $z$.
\end{definition}

Three properties together rule out every certificate the representation could support.
We prove this by exhibiting one such system, so the failure is a structural possibility, not an artifact of a poor encoder or an inaccurate model.

\begin{proposition}[Impossibility of transfer]\label{thm:nec}
There is a system with a bounded set of initial conditions, whose discarded direction is
\begin{enumerate}[label=\emph{(\roman*)},leftmargin=*,nosep]
\item unobservable from $z$,
\item unstable,
\item with true safety unbounded below along the reachable part of the fiber.
\end{enumerate}
For this system no certificate in the sense of \cref{def:cert} exists.
A trajectory with constant $z$ has $h(x(t))\to-\infty$ while remaining indistinguishable in $(z,u)$ from a safe one, so $\phi(z(\cdot))$ holds one value on both and cannot lower-bound $h$.
\end{proposition}

\begin{proof}
Take the system \eqref{eq:motivating-family} with $a<0$, $c=0$, safety $h=(1-x_1)-\ell\,x_2$ for a constant $\ell>0$, and the feedback $\pi\equiv0$, so $z=x_1$ is constant.
The discarded direction is the $x_2$-axis: unobservable, since $x_2$ never enters $\dot z=u$, and unstable, since $a<0$.
Take $x_1(0)=0$: the state $x_2^{(1)}(0)=0$ gives $h^{(1)}\equiv1$, safe for all time, while $x_2^{(2)}(0)=g_0>0$ gives $x_2^{(2)}(t)=g_0e^{-at}\to+\infty$ and $h^{(2)}(t)=1-\ell g_0e^{-at}\to-\infty$, with $z(\cdot)\equiv0$ on both.
So $h$ is unbounded below on $E^{-1}(0)\cap R$, giving \emph{(iii)}, and the two trajectories are indistinguishable, since the loop closes in $x_1$ alone.
Any certificate must satisfy $\bar\phi\le h$ along both, but $\phi(z(t))$ takes the same bounded value on the two, so it cannot lower-bound $h^{(2)}(t)\to-\infty$.
\end{proof}

To see that the third condition cannot be dropped, we remove it.
Replace $h$ by $(1-x_1)-\ell\tanh(x_2)$, bounded along the fiber.
The surrogate $\phi'(z)=(1-z)-\ell$ now satisfies $\bar\phi'\le h$, since $\tanh\le1$, so holding $\phi'\ge0$, that is $z\le1-\ell$, secures $h\ge0$.
A certificate exists even though the hidden direction is still unobservable and unstable.
Non-detectability makes failure possible, and unavoidable only when the direction is also unbounded below along the fiber.

\subsection{Invisibility to runtime monitors}\label{sec:detect}

In the system of \cref{thm:nec} the representation-only model $f_z(z,u)=u$ reproduces $\dot z$ exactly, so $\gamma=0$ and the latent CBF $\beta(z)=1-z$ satisfies the latent barrier inequality for every rate $\alpha$.
The transfer guarantee of \cite{lutkus2025latent} is not violated here, since its hypotheses are not met.
It requires the certified set to lie inside the induced safe set $\mathcal S_z$ of \cref{sec:safety}, and here every fiber holds an unsafe state, so $\mathcal S_z=\emptyset$ and the inclusion fails.
The missing hypothesis is the fiber gap this paper makes quantitative.
So $\beta=1-z$ is a valid CBF but not a certificate in the sense of \cref{def:cert}, since its lift overstates true safety, $\beta(E(x))-h(x)=\ell\,x_2>0$ for $x_2>0$, and \cref{thm:nec} says that no other choice of latent function repairs this.
No better monitor recovers what $\gamma$ misses, as the next corollary shows.

\begin{corollary}[Runtime impossibility]\label{cor:imp}
No causal function of $(z,u)_{[0,t]}$ can separate the failure of \cref{thm:nec}.
In particular, the residual is blind to it.
\end{corollary}

\begin{proof}
The safe and unsafe trajectories of \cref{thm:nec} produce identical $(z,u)_{[0,t]}$, so any causal function of $(z,u)$ takes the same value on both.
Since one is safe and the other has $h\to-\infty$, no such function separates them.
The residual is the special case $f_z(z,u)=u$, which reproduces $\dot z=u$ exactly, giving $r\equiv0$ and $\gamma=0$ throughout.
\end{proof}

The conjugacy error sees a hidden direction only as far as the representation does, so the gap can only be closed by changing what the encoder keeps.

\section{Transfer Guarantees}\label{sec:transfer}

Having shown when transfer is impossible, we now ask what a latent certificate that does carry over guarantees about the true system.
The guarantee is determined by the model error $\gamma$ and by the width of the fibers, where \cref{sec:limit} located the failure.
The model-error cost is known from \cite{lutkus2025latent} and \cref{sec:transfer-inv} restates it in our notation.
The fiber cost is left implicit in prior analyses, and \cref{sec:transfer-fiber} quantifies it.
Throughout this section the latent barrier has a bounded gradient over the reachable representations, $L:=\sup_{z\in E(R)}\|\nabla\phi(z)\|<\infty$, which holds in particular whenever $E(R)$ is bounded.

\subsection{The model-error margin}\label{sec:transfer-inv}

We take the latent certificate to be a control barrier function \cite{ames2019cbf} and $\pi$ a policy that enforces it on the latent model, so $\langle \nabla\phi(z), f_z(z,\pi(z))\rangle \ge -\alpha\,\phi(z)$
on $E(R)$ for a constant rate $\alpha > 0$.
Differentiating the lift $\bar\phi$ along the true dynamics and splitting off the residual of \cref{sec:safety} gives, with the gradient bound $L$,
\begin{equation}\label{eq:liftbarrier}
\tfrac{d}{dt}\,\bar\phi(x) \;\ge\; -\alpha\,\bar\phi(x) - L\gamma.
\end{equation}

The slack $L\gamma$ acts as a disturbance on safety, as in input-to-state safety \cite{kolathaya2019iss}, and it costs a margin.
Writing $\sigma_\gamma:=L\gamma/\alpha$, the set $\{x\in R:\bar\phi(x)\ge-\sigma_\gamma\}$ is forward invariant under the true dynamics, since $\tfrac{d}{dt}\bar\phi\ge0$ on its boundary.
The inequality and the invariant set are the barrier-transfer lemma and safety theorem of \cite[Lem.~3, Thm.~3]{lutkus2025latent}.
Their analysis stops at $\sigma_\gamma$ and absorbs the fiber gap into the worst-case induced set $(E(\mathcal{S}_x^c))^c$ \cite[Footnote~2]{lutkus2025latent}.

\subsection{The fiber-defect margin}\label{sec:transfer-fiber}

\Cref{sec:transfer-inv} bounds the lift, and the lift takes one value on each fiber.
True safety varies across the fiber, so the state can sit anywhere on it while $\bar\phi$ reports the same number.
A guarantee on $\bar\phi$ is therefore not yet a guarantee on $h$.
The gap is how far $h$ can fall across a fiber while $\bar\phi$ holds fixed, and we measure it directly in the values of $h$.
Fix a reference state $x_z\in E^{-1}(z)\cap R$ on each reachable fiber.
The fiber defect
\[
  D(z)\;:=\;\sup_{x\in E^{-1}(z)\cap R}\bigl(h(x_z)-h(x)\bigr)\;\ge\;0
\]
is the largest amount by which true safety on the reachable part of the fiber drops below its value at the reference.
The defect is one-sided because states safer than the reference are harmless, and it is finite exactly when $h$ is bounded below along the reachable part of the fiber.
An infinite defect is the unbounded safety-relevance of \cref{def:unbdd}.

\Cref{thm:transfer} turns the invariance of \cref{sec:transfer-inv} into a floor on true safety.
The cost is the defect on the current fiber, and its worst value over the reachable fibers makes the floor uniform.

\begin{theorem}[Transfer bound]\label{thm:transfer}
Suppose
\begin{enumerate}[label=\emph{(\roman*)},leftmargin=*,nosep]
\item $\phi$ satisfies the barrier inequality at rate $\alpha$, the gradient bound $L$ holds on $E(R)$, and $f_z$ is $\gamma$-forward-conjugate, so that $\{\bar\phi\ge-\sigma_\gamma\}$ is forward invariant with $\sigma_\gamma=L\gamma/\alpha$
\item the fiber defect is uniformly bounded on reachable fibers, $\bar D:=\sup_{z\in E(R)}D(z)<\infty$
\item the latent barrier under-approximates true safety at the reference state of each reachable fiber, $\phi(z)\le h(x_z)$.
\end{enumerate}
Then any trajectory starting in $\{\bar\phi\ge-\sigma_\gamma\}$ stays in it, and for all $t$,
\[
  h\bigl(x(t)\bigr)\;\ge\;\bar\phi\bigl(x(t)\bigr)-D\bigl(z(t)\bigr)\;\ge\;-\bigl(\sigma_\gamma+\bar D\bigr)
\]
\end{theorem}

\begin{proof}
Invariance of $\{\bar\phi\ge-\sigma_\gamma\}$ is the conclusion of \eqref{eq:liftbarrier}, so $\bar\phi(x(t))\ge-\sigma_\gamma$ for all $t$.
Fix $t$ and set $z=E(x(t))$.
The state $x(t)$ lies on the fiber over $z$ and inside $R$, so $h(x(t))\ge h(x_z)-D(z)$ by the definition of the defect.
By the calibration, $h(x_z)\ge\phi(z)=\bar\phi(x(t))$. Combining the three,
\[
  h(x(t))\;\ge\;\bar\phi(x(t))-D(z(t))\;\ge\;-(\sigma_\gamma+\bar D). \qedhere
\]
\end{proof}

We call a latent barrier that satisfies \emph{(i)} and \emph{(iii)} a \emph{calibrated barrier}.
The calibration $\phi(z)\le h(x_z)$ is not automatic, and without it the floor fails, so it is a condition the designer must enforce.
The defect needs no such care, since a fiber on which $h$ is constant contributes $D(z)=0$ and costs nothing.
With exact calibration $\phi(z)=h(x_z)$ the pointwise floor equals the right-hand side of \eqref{eq:fiber-inf} for every choice of reference, so it is the tightest value a certificate can take.
The decomposition adds two things to \eqref{eq:fiber-inf}.
For every choice of reference the defect is at most the variation of $h$ on the reachable fiber, $\sup h-\inf h$ over $E^{-1}(z)\cap R$, which depends on the encoder, the safety function and the reachable set alone, so it compares encoders before a barrier is chosen.
It can also be bounded from the dynamics without computing the infimum on each fiber (\cref{lem:stable}, \cref{prop:linear}).

\begin{remark}\label{rem:guarantee}
The floor is negative, so \cref{thm:transfer} is a robustness statement, not yet a safety guarantee.
It becomes one when the margin is built into the barrier the filter enforces. If the shifted barrier $\phi - (\sigma_\gamma + \bar D)$ satisfies the barrier inequality in its own right, \cref{sec:transfer-inv} applied to it, with the same gradient bound $L$, makes $\{\bar\phi - \sigma_\gamma - \bar D \ge -\sigma_\gamma\} = \{\bar\phi \ge \bar D\}$ forward invariant, and on it the floor reads $h \ge 0$.
The shift is not free, since the barrier inequality must be re-verified at the new boundary.
\end{remark}

\subsection{Bounding the defect}\label{sec:bounding}

The transfer bound is useful only when the fiber defect is finite, and \cref{thm:nec} showed that an unstable, unobservable direction with $h$ unbounded below drives it to infinity on a reachable fiber.
Hypotheses \emph{(ii)}--\emph{(iii)} below say the vertical subsystem, the dynamics of the discarded coordinates, is input-to-state stable with the observed component as input, which corresponds to output-to-state stability, the nonlinear notion of detectability of \cite{sontag1997output}.

Two coarser constants bound the defect from above.
The first is the along-fiber Lipschitz constant $L_V$ of $h$, the largest rate at which $h$ changes between states sharing a representation.
The along-fiber directions are the vertical subspace $\ker\partial E$ of the horizontal--vertical splitting of the tangent space \cite{friedl2026reduced}, the foliation structure underlying reduced-order models \cite{szalai2023data}.
The second is the fiber width $\rho$, the largest distance between two states on one fiber within the reachable set $R$.
Whenever both are finite, $D(z)\le L_V\rho$.
The product charges the metric width of the fiber where only the variation of $h$ matters, and a bounded $h$ on an unbounded fiber has $\rho=\infty$ yet a finite defect.

The coordinates $x=(z,v)$ are an assumption on the encoder.
Where $\partial E$ has full rank the submersion theorem straightens $E$ into the projection $(z,v)\mapsto z$ only locally, and \cref{lem:stable} assumes one such chart on the whole region of interest, as holds for a linear encoder.
The fibers of a general smooth encoder can be disconnected, which does not affect \cref{thm:transfer}, since the defect is taken over the whole reachable fiber.
Both $\rho$ and $L_V$ are measured in these coordinates.

\begin{lemma}[Contracting fibers]\label{lem:stable}
Suppose coordinates $x=(z,v)$ in which $E(x)=z$, so the fibers are the sets $\{z=\mathrm{const}\}$ and the vertical dynamics of the closed loop are $\dot v=f_v(z,v,u)$ with $u=\pi(z)$.
Assume, for all $z \in E(R)$, $u = \pi(z)$, and all $v, v'$,
\begin{enumerate}[label=\emph{(\roman*)},leftmargin=*,nosep]
\item the initial vertical components are bounded, $\|v(0)\|\le V_{0}$, and $h$ is $L_V$-Lipschitz along fibers
\item the vertical dynamics contract at a uniform rate $a>0$, $(v-v')^{\top}\bigl(f_v(z,v,u)-f_v(z,v',u)\bigr)\le-a\,\|v-v'\|^{2}$
\item the drive is bounded, $\|f_v(z,0,u)\|\le\kappa$
\end{enumerate}
Then every reachable state satisfies $\|v\|\le V_{0}+\kappa/a$, the fiber width within $R$ is finite, $\rho\le2\,(V_{0}+\kappa/a)$, and
\[
  \bar D\;\le\;L_V\rho\;<\;\infty ,
\]
so hypothesis \emph{(ii)} of \cref{thm:transfer} holds.
\end{lemma}

\begin{proof}
Along any closed-loop trajectory, \emph{(ii)} with $v'=0$ and \emph{(iii)} with Cauchy--Schwarz give $\tfrac{d}{dt}\tfrac12\|v\|^{2}=v^{\top}f_v(z,v,u)\le-a\,\|v\|^{2}+\kappa\,\|v\|$, hence by the comparison lemma $\|v(t)\|\le e^{-at}\|v(0)\|+\kappa/a\le V_{0}+\kappa/a$.
Two reachable states on one fiber share $z$ and differ only in $v$, so their distance is at most $2\,(V_{0}+\kappa/a)$, and the Lipschitz bound along fibers converts the radius into the defect, $D(z)\le L_V\rho$ on every reachable fiber.
\end{proof}

\subsection{The linear case}\label{sec:lti}

The impossibility of \cref{thm:nec} and the contracting sufficiency of \cref{lem:stable} leave a gap.
A direction invisible to the encoder now may become visible later through the dynamics.
For linear systems $\dot x=Ax+Bu$ with $z=Cx$ the gap closes.
The unobservable subspace $N=\bigcap_{k}\ker(CA^{k})$ is $A$-invariant, so a direction in $N$ stays unobservable for all time.
On $N$ the controller plays no role, since states sharing a representation receive the same input, and their difference evolves open loop, $\dot d=Ad$.
Only $C$, $A$, and the safety gradient therefore matter for a certificate to survive the encoder.

\begin{definition}[Safety-detectability]\label{def:safety-detect}
Write $(C,A)$ in observability form, $x=(\xi,\eta)$ with $z=C_{1}\xi$ and $\eta$ coordinates on $N$. Let the unobservable block $A_{\bar o}$ have no eigenvalues on the imaginary axis, so that $\eta$ splits into modal coordinates $(\eta_s,\eta_u)$ with $A_s$ Hurwitz and $A_u$ anti-Hurwitz, and $N_u\subseteq N$ is the subspace of the $\eta_u$ coordinates.

Let $h$ be affine with gradient $w=(w_\xi,w_s,w_u)$ in these coordinates.
The safety function is \emph{safety-detectable} through the encoder if $w_u=0$.
Every mode the encoder cannot see and the dynamics do not contract is then one that true safety does not depend on.

\end{definition}

When $A_{\bar o}$ has no eigenvalues on the imaginary axis, safety-detectability corresponds to functional detectability of the safety output $y_h=w^\top x$ from the measured output $z$ \cite{fernando2010functional}.
Checking it takes one observability decomposition and one orthogonality test, $w^{\top}$ applied to a basis of $N_u$, which is available before a policy or a barrier is trained.

\begin{proposition}[Linear characterization]\label{prop:linear}
Let $\dot x=Ax+Bu$ with $z=Cx$ and $u=\pi(z)$ for a locally Lipschitz static $\pi$. Let $h$ be affine, let $A_{\bar o}$ have no eigenvalues on the imaginary axis, and suppose the observable component of the state is bounded on $R$ and the initial set is bounded.
\begin{enumerate}
\item[(i)] If $h$ is safety-detectable, then $\bar D<\infty$ for every choice of reference states, and any latent barrier satisfying the conditions of \cref{thm:transfer} transfers with floor $-(\sigma_\gamma+\bar D)$.
\item[(ii)] If $h$ is not safety-detectable, and the initial set contains a ball around a point whose trajectory keeps $h$ bounded above, then no continuous $\phi$ satisfies \cref{def:cert}, and no continuous $\phi$ calibrated as in \cref{thm:transfer} has $\bar D<\infty$.
\end{enumerate}
\end{proposition}

\begin{proof}
(i) In the coordinates of \cref{def:safety-detect} the observable component $\xi$ is bounded on $R$ by hypothesis, so $z$ is bounded and $u=\pi(z)$ is bounded by continuity.
The block $A_s$ is Hurwitz and driven by $\xi$ and $u$ alone, the modal splitting decouples it from $\eta_u$, so $\eta_s$ is bounded on $R$ as well.
Two reachable states on one fiber differ by a vector whose $\xi$- and $\eta_s$-components are differences of bounded quantities and whose $\eta_u$-component is annihilated by $w$.
So $|h(x)-h(x')|$ is uniformly bounded over reachable pairs sharing a representation, every choice of reference states gives $\bar D<\infty$, and \cref{thm:transfer} yields the floor.

(ii) Suppose a certificate $\phi$ exists.
Let $x_0$ be the point of the hypothesis, $x(\cdot)$ its trajectory, $h(x(t))\le H$, and $z(\cdot)$ confined to a compact set $K$ by the boundedness of the observable component.
Since $w_u\neq0$ and $e^{-At}$ contracts $N_u$ exponentially, $\|w^{\top}e^{At}|_{N_u}\|$ grows without bound, and by finite dimensionality some $d\in N_u$ has $\sup_t|w^{\top}e^{At}d|=\infty$: were $t\mapsto w^{\top}e^{At}d$ bounded for each $d$ in a basis, the operator norm would be bounded too.
For $\varepsilon$ small enough the companions $x_{\pm}(t)=x(t)\pm\varepsilon\,e^{At}d$ are closed-loop trajectories from $x_0\pm\varepsilon d$ in the initial set: $N$ is $A$-invariant and contained in $\ker C$, so $Ce^{At}d\equiv0$, the three trajectories share the representation $z(\cdot)$ and therefore the input, and each difference obeys the open-loop dynamics.
Any certificate must then satisfy, for all $t$,
\[
  \phi\bigl(z(t)\bigr)\;\le\;\min_{\pm}\,h\bigl(x_{\pm}(t)\bigr)
  \;=\;h\bigl(x(t)\bigr)-\varepsilon\,\bigl|w^{\top}e^{At}d\bigr|,
\]
which is at most $H-\varepsilon\,|w^{\top}e^{At}d|$ and hence unbounded below along a sequence of times, while a continuous $\phi$ is bounded on the compact set $K$.
No such $\phi$ exists.
If a continuous $\phi$ calibrated at reference states had $\bar D<\infty$, the pointwise floor of \cref{thm:transfer} would make $\phi-\bar D$ a certificate, which was just excluded.
\end{proof}

The failure in (ii) lies in the encoder and not in the policy.
A causal policy $u=\pi(z_{[0,t]})$ receives identical histories on the companion trajectories, so their difference still evolves open loop and the proof applies unchanged.
The condition itself is therefore the same for every causal policy of $z$, since for linear systems an input distinguishes two states only if every input does \cite{hermann1977nonlinear}, while the size of $\bar D$ and the boundedness of the observable component depend on the closed loop.
The imaginary-axis exclusion matters.
A discarded integrator grows linearly under an observed drive of nonzero mean although every autonomous solution is bounded, as the hidden mode of \cref{sec:exp-synthetic} does at $a=0$.
On the axis the spectrum no longer decides, and \cref{def:unbdd} is applied to the reachable fiber directly.
Safety-detectability is strictly weaker than detectability of $(C,A)$, which demands $N_u=0$ outright.
\Cref{lem:stable} gives a nonlinear sufficient condition for a uniformly bounded defect.
The bounded-$h$ hypothesis of (ii) is a feasibility condition, ruling out the degenerate loop whose every trajectory becomes unboundedly safe.

\section{Numerical Examples}\label{sec:experiments}

\subsection{The synthetic system}\label{sec:exp-synthetic}

We re-use system \eqref{eq:motivating-family} and fix $b = 1$, $\ell = 0.5$, the zero input $u \equiv 0$, the exact latent model and the latent CBF $\phi(z) = 1-z$.
The initial set is $x_1(0)=z_0=0.5$ and $x_2(0)\in[-0.5,0.8]$, so the entire history remains on the fiber over $z_0$ and the latent barrier reports $\phi(z_0)=0.5$ throughout.
\Cref{fig:sweep} sweeps the hidden-state stability parameter $a$.
For $a>0$ the vertical dynamics contract at rate $a$, so \cref{lem:stable} applies and \cref{thm:transfer} gives a finite transfer floor.
The steady state $x_2=bz_0/a$ lies above the initial set for $a<0.625$, so $\bar D(a)=\ell\max\{0.8,\,bz_0/a\}$ and the uniform floor diverges as $a\to0^+$.
For $a<0$ the hidden direction is unstable and \cref{thm:nec} rules out a certificate, while at $a=0$ the constant drive produces an unbounded reachable fiber directly.
The latent barrier is insensitive to this transition: it remains at $0.5$, while true worst-case safety crosses zero at $a=0.5$.
For $a=0.8$ and reference state $x_{z_0}=(z_0,0)$ the largest reachable hidden coordinate is $x_2=0.8$, so $D(z_0)=0.4$ and the pointwise floor $\bar\phi-D(z_0)=0.1$ is attained at $t=0$.

\begin{figure}[!t]
\centering
\includegraphics[width=\columnwidth]{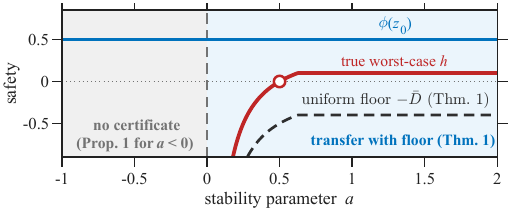}
\caption{Stability sweep on one fiber of \eqref{eq:motivating-family} with $u\equiv0$ and exact latent model: the barrier $\phi(z_0)$, the true worst-case $h$ and the uniform floor of \cref{thm:transfer}.}
\label{fig:sweep}
\end{figure}

\subsection{A cartpole encoder}\label{sec:exp-cartpole}

The plant is the standard frictionless cartpole with state $x=(p,\dot p,\theta,\dot\theta)$ and $\theta$ measured from upright.
The cart and pole masses are $0.4\,\mathrm{kg}$ and $0.1\,\mathrm{kg}$, the pole length is $0.25\,\mathrm{m}$, and the input is limited to $|u|\le5\,\mathrm{N}$.
Safety is defined by the $2\,\mathrm{m}$ track half-width, $h(x)=1-|p|/2$, and the expert controller is a full-state LQR gain ($Q=\mathrm{diag}(1,0.5,5,1)$, $R=1$).
A pole-only encoder $E_2(x)=(\theta,\dot\theta)$ puts this plant in the no-certificate case of \cref{fig:triage}.
For the frictionless dynamics $\ddot\theta$ depends on $(\theta,\dot\theta,u)$ but not on $(p,\dot p)$, so the loop closes through this representation with the pole part of the expert gain and the cart gains removed.
The cart block of the closed-loop Jacobian at upright is a double integrator, on the imaginary axis excluded by \cref{prop:linear}, but \cref{def:unbdd} applies directly: an initial state with the pole upright and a nonzero cart velocity keeps $z\equiv0$ and $u\equiv0$ while $p$ grows without bound, so $h$ is unbounded below on the reachable fiber over $z=0$ and no certificate exists.
Two runs from $(\theta,\dot\theta)(0)=(0.1\,\mathrm{rad},0)$ with hidden cart states $(p,\dot p)(0)=(0,0)$ and $(1.5\,\mathrm{m},0.5\,\mathrm{m/s})$ have identical latent and residual histories, yet their carts cross opposite track boundaries at $9.0$ and $1.8\,\mathrm{s}$.
As in \cref{cor:imp}, the failure is invisible to any causal monitor of $(z,u)$, and discarding the cart position, on which safety depends, costs nothing in latent-model accuracy.

\subsection{Learned encoders}\label{sec:exp-learned}

We now estimate the two margins of the transfer-with-floor case of \cref{fig:triage} on two learned $4$--$16$--$2$ tanh encoders of the cartpole of \cref{sec:exp-cartpole} that differ only by a safety-readout loss.
The example illustrates how the two margins can be estimated for a learned encoder and is not a certification procedure.
With samples $(x_i,\dot x_i,u_i)$ of the expert, $z_i=E(x_i)$ and $p_i$ the cart position, they minimize
\begin{equation}\label{eq:loss}
\begin{split}
J=\frac{1}{N}\sum_{i=1}^{N}\Big(&
\underbrace{\|\partial E(x_i)\dot x_i-A_z z_i-B_z u_i\|^2}_{\text{latent dynamics}}\\[-1mm]
&+w_\pi\underbrace{(K_z z_i-u_i)^2}_{\text{policy imitation}}
+w_s\underbrace{(R_z z_i-p_i)^2}_{\text{safety readout}}\Big)
\end{split}
\end{equation}
The readout asks only that a safety-relevant signal, here $p$, be recoverable from $z$.
We use $w_\pi=1$, $w_s=10$ for the safety-readout and $w_s=0$ for the plain variant, and train both on the same $40$ closed-loop runs of $8\,\mathrm{s}$ of the expert of \cref{sec:exp-cartpole} under $1\,\mathrm{N}$ exploration noise, sampled at $50\,\mathrm{Hz}$, with $30$ matched seeds per variant.
The linear maps $A_z$, $B_z$, $K_z$, $R_z$ are refit by least squares at every epoch for both variants, so for $w_s=0$ the readout is a post-hoc fit.

On $10$ held-out runs both variants imitate the expert at the $1\,\mathrm{N}^2$ floor set by the exploration noise ($0.999\pm0.022$ and $1.000\pm0.023\,\mathrm{N}^2$), with latent-dynamics losses of $0.0035\pm0.0023$ and $0.0054\pm0.0015$.
The safety-readout RMSE falls from $0.139\pm0.016\,\mathrm{m}$ to $0.0060\pm0.0015\,\mathrm{m}$ (\cref{fig:learned}(a)).
Every deployed loop is closed by the latent policy $u=K_z z$ of its encoder.
A seed is deployment-valid if its loop stays inside the track and pole-angle region from every evaluation start over the full horizon, so every deployment-valid loop is safe.
One plain encoder leaves the track and is excluded, and all $30$ safety-readout encoders are valid.

The measured latent-model error $\hat\gamma$, which determines the model-error contribution $\sigma_\gamma$, does not point to the safety-readout encoders.
Its ranges overlap, $0.23$ to $2.27$ for the plain and $0.38$ to $3.49$ for the safety-readout encoders, and it is larger on average for the latter, $1.29$ against $0.96$ (rank-sum $p=0.08$, \cref{fig:learned}(b)).
The resulting margin is $\sigma_\gamma=0.195\pm0.151$ for the plain and $0.345\pm0.169$ for the safety-readout encoders.
The calibration below holds only at the visited samples, while the defect of \cref{thm:transfer} covers the whole reachable fiber, so we charge it separately as the variation of $h$ on the fibers, which bounds the defect for every reference, and estimate it in two ways, keeping the larger as $\hat{\bar D}$.
The first is the largest drop of $h$ between visited states whose representations differ by less than $0.05$, reduced heuristically by $0.05\,L$ for this tolerance, about $0.04$ and $0.05$ at the median $\|R_z\|$ of $1.6$ and $2.1$ of the two variants. It stays inside $R$ but compares neighboring fibers and not one fiber.
The second is a descent of $h$ along one connected component of the fiber inside the bounding box of the visited states, which can miss other components and can leave $R$.
Both can understate or overstate $\bar D$, so the estimate serves diagnosis and not certification.
It does distinguish the variants, with a median falling from $0.803$ to $0.032$ and no overlap between them (\cref{fig:learned}(c)).
The first estimate alone gives a plain median of $0.71$, and its worst pairs join a visited state near the track center with one near $\pm1.5\,\mathrm{m}$ at a latent distance below $0.05$, so large variations already occur between visited states.

For each valid deployment the latent barrier is built from the safety readout, $\phi(z)=1-|R_z z|/2-\delta$, where the calibration shift $\delta$ is the smallest constant giving $\phi(E(x))\le h(x)$ at every visited sample, with gradient bound $L=\|R_z\|/2$ and rate $\alpha=4$, a design choice whose effect is reported below.
For this example we relax the smoothness of \cref{sec:safety} to the locally Lipschitz minimum of the two affine barriers $1\mp R_z z/2-\delta$: if both satisfy \eqref{eq:liftbarrier}, their minimum satisfies it almost everywhere along trajectories, which is all the comparison argument of \cref{sec:transfer-inv} uses.
At every visited sample both satisfy the barrier inequality on the learned latent model for all $30$ safety-readout encoders, with smallest margin $0.372$, while it fails for $14$ of the $29$ plain encoders, to which \cref{thm:transfer} then does not apply.

We then compare the budget of the calibrated barrier, its largest value $1-\delta$, with the empirical cost $\sigma_\gamma+\hat{\bar D}$ of \cref{thm:transfer}.
The corridor is the set $\{\phi\ge\sigma_\gamma+\hat{\bar D}\}$. With certified bounds and calibration, \cref{thm:transfer} would give $h\ge0$ there.
In the read-out cart position $\hat p=R_z z$ the corridor is $|\hat p|\le 2\,(1-\delta-\sigma_\gamma-\hat{\bar D})$.
It is nonempty when the budget exceeds the cost, and only then do we call the estimated bound non-vacuous.
The estimates $\hat\gamma$, $\hat{\bar D}$ and the calibration all come from samples, so a nonempty corridor does not certify $h\ge0$, and its invariance is not established, since the deployed policy is not a filter for the shifted barrier of \cref{rem:guarantee}.
Its inequality holds at the visited samples inside the corridor and fails outside, mostly on runs that start outside, and once inside every later sample of every run stays in the set $\{\bar\phi\ge\hat{\bar D}\}$ of \cref{rem:guarantee}.

At $\alpha=4$ none of the $29$ deployment-valid plain encoders yields a nonempty empirical corridor, whereas $29$ of $30$ safety-readout encoders do (\cref{fig:learned}(d)), and the counts are the same with only the pair estimate or only the descent in place of $\hat{\bar D}$.
The calibration already accounts for the variation of $h$ among the visited states themselves, so on the samples this charge is conservative. How much of it the unvisited part of the fibers requires is not determined by the samples, and we keep it in full.
The rate enters the barrier inequality and the margin $\sigma_\gamma=L\gamma/\alpha$, and leaves the defect untouched.
At $\alpha=1,2,4,8$ all $30$ safety-readout encoders pass the barrier check, with smallest margins $0.156$, $0.228$, $0.372$ and $0.619$, their corridor count is $3$, $19$, $29$ and $30$ of $30$, and the plain count is $0$ of $29$ at every rate.
The $59$ valid loops keep the cart on the track, so the split does not measure safe behavior.
It measures whether the estimated cost of \cref{thm:transfer} leaves any part of the barrier's range, and for the plain encoders it does not, so under this conservative charge the bound certifies $h\ge0$ nowhere.
The excluded plain encoder gets the same verdict from its four runs that stay on the track, an estimated defect of $0.75$ against a budget of $0.28$, and with that calibration held fixed its fifth run reaches $h=0$ while the barrier still reads $0.26$.

\begin{figure}[!t]
\centering
\includegraphics[width=\columnwidth]{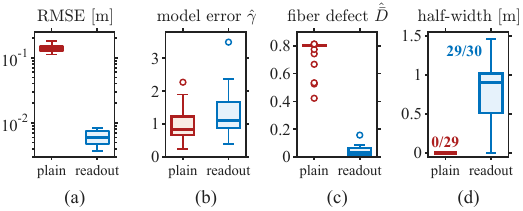}
\caption{Learned encoders over $30$ matched seeds: (a)~safety-readout RMSE, (b)~latent-model error $\hat\gamma$, (c)~empirical fiber defect $\hat{\bar D}$, (d)~half-width of the empirical transfer corridor at $\alpha=4$, zero when the corridor is empty, nonempty for $0$ of $29$ plain and $29$ of $30$ safety-readout encoders.}
\label{fig:learned}
\end{figure}

\section{Conclusion}\label{sec:conc}

Detectability of the discarded safety-relevant information governs whether latent transfer succeeds.
A direction that is unobservable, unstable, and unbounded below along the fiber admits no certificate, and no monitor downstream of the encoder can see it, while a contracting discarded direction admits the transfer bound.
In the linear case the boundary is exact under the hypotheses of \cref{prop:linear}: with safety-detectability every calibrated latent barrier function transfers with a finite margin, and without it none does and no continuous certificate exists.

Several limitations remain.
No example synthesizes a filter and the experiments stay on a cartpole with low-dimensional encoders that take the physical state as input and are trained on ground-truth states, their derivatives, and a full-state expert, information that the image-based world models motivating this work do not have.
The analysis assumes a static, continuously differentiable encoder of the state, and its extension to recurrent and discrete world-model encoders remains open.
\Cref{thm:nec} needs a direction that is exactly hidden, and there memory of the $z$-history cannot help.
When a learned encoder leaks the hidden direction into $z$, recovery becomes an observer problem, and how its gain, time window and noise sensitivity scale as the leakage vanishes remains open.

In the experiments $59$ of the $60$ encoders kept the cart on the track, and among them the latent-model error did not indicate which admit a non-vacuous estimated bound, while the estimated fiber defect did.
A safety-readout loss reduced the estimated fiber defect by more than one order of magnitude without degrading policy imitation.
All estimates are obtained from closed-loop samples of the physical state, so they do not certify the safety-readout encoders, and an estimated $h$ adds a margin the certificate does not charge.
They indicate at design time whether the estimated bound of a given encoder is vacuous, which the latent signals cannot reveal.

\bibliographystyle{IEEEtran}
\bibliography{references}

\end{document}